\documentclass[preprint,11pt]{elsarticle}
\usepackage{amssymb,amsthm,amsmath}
\usepackage{bm}
\usepackage{booktabs}
\usepackage{xcolor}
\usepackage{url}

\usepackage{tikz,amssymb}
\usetikzlibrary{shapes,positioning,calc}
\tikzstyle{tre}=[circle,draw,minimum size=1.5mm]
\tikzstyle{small}=[circle,draw,inner sep=0pt, minimum size=1mm]

\newtheorem{thm}{Theorem} 
\newtheorem{prop}[thm]{Proposition} 
\newtheorem{lem}[thm]{Lemma} 
\newtheorem{coro}[thm]{Corollary} 
\newdefinition{rmk}{Remark}
\newdefinition{df}{Definition}

\newcommand{\bmu}{\bm{\mu}}
\DeclareMathOperator{\indeg}{indeg}
\DeclareMathOperator{\outdeg}{outdeg}

\newcommand{\arc}{}

\newcommand{\defit}[2][]{%
  \emph{#2}%
}

\title{Comparison of orchard networks using their extended $\mu$-representation\tnoteref{t1}}

\tnotetext[t1]{Supported by Grant PID2021-126114NB-C44 funded by MCIN/AEI/ 10.13039/501100011033 and by ``ERDF A way of making Europe''.}

\author[1]{Gabriel Cardona}
\ead{gabriel.cardona@uib.es}
\author[1]{Joan Carles Pons\corref{cor1}}
\ead{joancarles.pons@uib.es}
\author[2]{Gerard Ribas}
\ead{gerard.ribas1@estudiant.uib.cat}
\author[1]{Tomás Martínez Coronado}
\ead{t.martinez@uib.eu}

\cortext[cor1]{Corresponding author}

\affiliation[1]{
    organization={Department of Mathematics and Computer Science -- University of the Balearic Islands},
    addressline={Ctra. Valldemossa, km. 7.5},
    postcode={ES-07122},
    city={Palma},
    country={Spain}}
\affiliation[2]{
    organization={Higher Polytechnic School -- University of the Balearic Islands},
    addressline={Ctra. Valldemossa, km. 7.5},
    postcode={ES-07122},
    city={Palma},
    country={Spain}}

\journal{arXiv}

\begin{document}

\begin{frontmatter}

\begin{abstract}
    Phylogenetic networks generalize phylogenetic trees in order to model reticulation events. Although the comparison of phylogenetic trees is well studied, and there are multiple ways to do it in an efficient way, the situation is much different for phylogenetic networks.

    Some classes of phylogenetic networks, mainly tree-child networks, are known to be classified efficiently by their $\mu$-representation, which essentially counts, for every node, the number of paths to each leaf.
    In this paper, we introduce the \emph{extended} $\mu$-representation of networks, where the number of paths to reticulations is also taken into account. This modification allows us to distinguish orchard networks and to define a sound metric on the space of such networks that can, moreover, be computed efficiently. 

    The class of orchard networks, as well as being one of the classes with biological significance (one such network can be interpreted as a tree with extra arcs involving coexisting organisms), is one of the most generic ones (in mathematical terms) for which such a representation can (conjecturally) exist, since a slight relaxation of the definition leads to a problem that is Graph Isomorphism Complete.
\end{abstract}


\end{frontmatter}

\section*{Introduction}


Phylogenetic networks generalize phylogenetic trees to model complex evolutionary relationships that are not well fitted by trees. Namely, phylogenetic networks can have \emph{reticulations}, nodes with multiple incoming arcs, that model interactions between existing operational taxonomic units (or OTUs, for short) giving rise to a new one. These interactions, which can be recombinations, hybridizations or lateral gene transfers (among others), are increasingly evident in organisms such as plants, bacteria and viruses, highlighting the limitations of tree representations~\cite{martin1999mosaic, linder2004reconstructing}. In this paper, we shall only consider \emph{binary} phylogenetic networks. This rules out the possibility that there are uncertainties about the order in which speciation events occur (``soft'' polytomies, in biological terms) and also that more than two different OTUs interact in a reticulation event.

The study of phylogenetic networks is a growing and very active field of research. However, working with phylogenetic networks, in its broadest sense, is not appropriate if the aim is to give objects that can be computationally tractable or to represent biologically meaningful scenarios. Then, pursuing these objectives,  multiple subclasses of phylogenetic networks have appeared in the literature during the last years (see~\cite{kong2022classes} for a recent review). In this paper, we shall focus on the subclass of orchard (phylogenetic) networks \cite{erdos2019class}, also called  cherry-picking networks \cite{janssen2021cherry}, that fulfill the conditions that they represent biologically significant scenarios and also form a computational and mathematical prominent class of networks.

Biologically, orchard networks have a natural  interpretation, since such a network can be understood as a  tree with \emph{horizontal arcs} (that is, involving coexistent OTUs) representing reticulate events \cite{van2022horizontal}. See Definition~6 in \cite{van2022horizontal} for a formal description of this type of time-consistent labeling (called HGT-consistent labeling therein), and note that it is similar but different from other definitions of time-consistency used for a long time in a similar context, like in \cite{moret2004phylogenetic} for instance. This biological interpretation also evidences that orchard networks are a subclass of tree-based networks \cite{francis2015phylogenetic}, where the additional arcs do not need to be horizontal. 
Mathematically, orchard networks are those that can be reduced to a trivial network by the iterative application of reductions, each of which removes certain atomic structures called \emph{cherries} and \emph{reticulated-cherries}. Note also that orchard networks include \emph{tree-child} networks~\cite{cardona2008comparison}, one of the most studied class of networks.

In this paper we shall consider the problem of the comparison of orchard networks. First, in order to detect if two (in principle) different networks are the same (more formally, \emph{isomorphic}). Second, and related to that, to quantify the difference between these two networks (more formally, compute some \emph{distance} between them). The comparison of networks is needed in practical applications, since different networks may arise when using different reconstruction methods or when using different samples of DNA. The usual strategy to make this comparison is to associate to each network some invariants of its isomorphism class that are easy to compare, and make this comparison in the space of invariants. 
Since the problem of deciding when two generic phylogenetic networks are isomorphic is Graph Isomorphism Complete~\cite{cardona_comparison_2014}, it is believed that there cannot exist invariants that can both distinguish arbitrary phylogenetic networks and be computed efficiently.
Hence, for each set of invariants there is a certain class of phylogenetic networks on which it separates networks (that is, where different networks have different invariants). These invariants can be either substructures of a certain kind (simpler than generic networks) included in the network  (like trees~\cite{willson2010regular} or triplets~\cite{Gambette2012,van2014trinets,cardona2009metrics2,Cardona2011,semple2021trinets}), distances between pairs of leaves~\cite{cardona2009metrics2,cardona2010path,Bordewich2016}, or certain data associated to the nodes of the network~\cite{cardona2008comparison,moret2004phylogenetic,erdos2019class}.

One of the invariants that fall into this last class is the \emph{$\mu$-representation} of the network, which is known to separate tree-child~\cite{cardona2008comparison} and semi-binary tree-sibling time-consistent~\cite{cardona2008distance} networks. This invariant allows us to define the \emph{$\mu$-distance} between networks, a sound distance in the space of (isomorphism classes of) networks of the aforementioned classes over a fixed set of taxa, that generalizes to networks the well known Robinson-Foulds distance on phylogenetic trees. The $\mu$-representation of a network is the (multi)set of the $\mu$-vectors of its nodes, which in turn count the number of paths from the node under consideration to each of the leaves. This data is essentially the same as the \emph{ancestral profile} of the network, which classifies \emph{stack-free} (that is, with no pair of reticulations connected by an arc) orchard networks~\cite{Bai2021}, but not generic orchard networks (notice that in~\cite{Bai2021} the authors give a counterexample to \cite[Theorem~2.2]{erdos2019class}).


In this paper we show how a modification of the $\mu$-vectors (namely, counting also the overall number of paths from the node to any reticulation) allows to separate (binary) orchard networks with no further restrictions.

The paper is organized as follows. In Section~\ref{sec:preliminaries} we give the basic definitions used throughout the manuscript, and in Section~\ref{sec:orchard} we review the definition of orchard networks. In Section~\ref{sec:exended-mu} we define formally the (extended) $\mu$-representation associated to a network, and in Section~\ref{sec:mu-for-orchard} we prove that it is enough to classify orchard networks and we use it to define a sound metric distance on the class of such networks.
Finally, in Section~\ref{sec:conclusions} we address some final remarks and show an implementation of the results of the paper.

\section{Preliminaries}
\label{sec:preliminaries}

In this section we give some definitions that will be used throughout the manuscript. 

For any positive integer $n$, we denote by $[n]$ the set $\{1,\dots,n\}$ and by $[n]^*$ the set $\{0,1,\dots,n\}$.

We shall be working with directed acyclic graphs $N=(V,A)$.
Given a node $u\in V$, we denote by $\indeg u$ (resp. $\outdeg u$) the number of arcs whose head (resp. tail) is $u$.
Given two nodes $u,v\in V$, if there is an arc with tail $u$ and head $v$  (or from $u$ to $v$), in symbols $uv\in A$, we say that $u$ is a parent of $v$ or that $v$ is a child of $u$. If there is a directed path from $u$ to $v$ we say that $u$ is an ancestor of $v$, or that $v$ is a descendant of $u$. Note that it includes the case where $u=v$, since we allow trivial paths, of length $0$. 

We denote by $M_N(u,v)$ the set of paths in $N$ from $u$ to $v$ and by $m_N(u,v)=|M_N(u,v)|$ the number of such paths; in particular, $m_N(u,u)=1$ for every node $u$ of $N$. If the graph is clear from the context, we shall omit the subscript and simply write $M(u,v)$ and $m(u,v)$.

We say that a node $u$ in a directed graph is \defit[node]{elementary} if $\indeg u=\outdeg u=1$. Its \defit{simplification} consists in removing it (together with its incident arcs) and connecting its single parent to its single child.

    A \defit{phylogenetic network}, or simply a \defit{network}, $N=(V,A,\varphi)$
    on a set $X$ of \emph{taxa}, is a directed acyclic graph $(V,A)$ without parallel arcs such that any node $u\in V$ is either:
    \begin{enumerate}[(i)]
        \item a \defit{root}, with $\indeg u=0$, $\outdeg u=1$ (and there can only be one such node), or
        \item a \defit{leaf}, with $\indeg u=1$, $\outdeg u=0$, or
        \item a \defit{tree node}, with $\indeg u=1$, $\outdeg u=2$, or
        \item a \defit{reticulation}, with $\indeg u=2$, $\outdeg u=1$,
    \end{enumerate}
    together with a \emph{labeling} function $\varphi$ which is a bijection between $X$ and the set of leaves.


A \defit{phylogenetic tree} (or simply a \defit{tree}) is a phylogenetic network without reticulations.

We shall hereafter identify the set of taxa and the set of leaves, and assume that $X\subseteq V$, which means that we can drop $\varphi$ from the definition of network.
Given a phylogenetic network $N$, we shall denote by $V_H(N)$ the set of its reticulations and by $V_T(N)$ the set of nodes that are either leaves or tree nodes. 
Moreover, we denote by $\rho(N)$ the root of $N$. In case the network is clear from the context, we may simply write $V_H$, $V_T$ and $\rho$, respectively.

In case that $X$ is a singleton, $X=\{x\}$, we shall indicate by $I_x$ the network whose only nodes are its root and the leaf $x$, joined by an arc.

From now on, we shall always assume that $X$ is formed by positive integers, and hence $X\subseteq [n]$ for some $n$. 


Finally, as it is usual in phylogenetics, we shall say that two networks $N_1=(V_1,A_1)$ and $N_2=(V_2,A_2)$ on $X$ are \defit[networks]{isomorphic}, in symbols $N_1\cong N_2$, if there exists a bijection $f: V_1 \rightarrow V_2$ such  
that is the identity on $X$, and such that
 $uv \in A_1$ if and only if $f(u)f(v) \in A_2$.

\section{Orchard networks}
\label{sec:orchard}


In this section, we summarize some definitions and results from~\cite{janssen2021cherry}, but notice that some of the notations that we use are taken from~\cite{erdos2019class}, and some others are new.

Let $N=(V,A)$ be a network on $X\subseteq[n]$ and let $(i,j)\in X\times X$ with $i\neq j$. Also, denote by $p_i, p_j$ the parents of the leaves $i$ and $j$ in $N$, respectively. We call $(i,j)$ a \defit{cherry} if $p_i=p_j$, and we call it a \defit{reticulated-cherry} if $p_i$ is a reticulation, $p_j$ is a tree node, and $p_j$ is one of the parents of $p_i$. In either case $(i,j)$ is a cherry or a reticulated-cherry, we say $(i,j)$ is \defit{reducible} in $N$.

If $(i,j)$ is reducible in $N$, the \defit{reduction} of $(i,j)$ in $N$, denoted by $N^{(i,j)}$, is the result of:
\begin{itemize}
    \item If $(i,j)$ is a cherry, then 
    remove the leaf $i$ (and its incoming arc) and simplify $p_i$, which is now an elementary node.
    \item If $(i,j)$ is a reticulated-cherry, then delete the arc $p_jp_i$ and then simplify $p_i$ and $p_j$, which are now elementary nodes. 
\end{itemize}


Notice that the conditions of being a cherry and a reticulated-cherry are clearly incompatible, which implies that the operation is well defined.

Given a sequence of pairs of integers $S=(s_1,\dots,s_k)$ which, for brevity, we will write as $S=s_1\cdots s_k$, with $s_t=(i_t,j_t)$ and $i_t,j_t\in[n]$, of length $k\ge1$, we say that $S$ is \defit{reducible} in $N$ if:
\begin{itemize}
    \item $s_1$ is reducible in  $N$.
    \item For every $t\in\{ 2,\dots,k \}$, $s_t$ is reducible in
    $(\dots(N^{s_1})^{s_2}\dots)^{s_{t-1}}$.
\end{itemize}
In such a case, we shall define the \defit{reduction} of $N$ with respect to $S$ as $(\dots(N^{s_1})^{s_2}\dots)^{s_{k}}$ and it will be denoted by $N^S$.

Moreover, we say that $S$ is \defit{complete} if $N^S=I_i$ for some $i\in X$ and, in case that one such complete sequence exists, we call $N$ an \defit{orchard network}.


The fundamental result that allows to classify orchard networks using complete reducible sequences is the following, which is adapted from~\cite[Corollary~1]{janssen2021cherry}.

\begin{thm}
    \label{same-crs-implies-iso}
    Let $S$ be a complete reduction sequence for two orchard networks $N$ and $N'$. Then, $N\cong N'$. 
\end{thm}

Notice, however, that the complete reduction sequence for an orchard network is not, in general, unique. 

Figure~\ref{fig:definicions} shows an orchard network together with a sequence of reductions that reduce it to a trivial network. Notice that we can take the trivial network, invert the reductions that have been applied, and recover the initial network, using only the information of the sequence of reductions that have been applied. This is essentially the proof of Theorem~\ref{same-crs-implies-iso}.

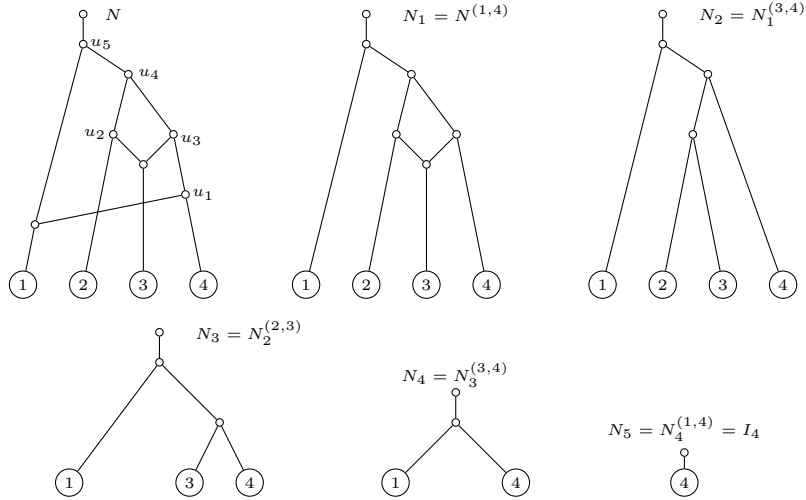
\begin{figure}[t]
\centering
\begin{tabular}{@{}c@{}}
\begin{tikzpicture}[scale=0.8]
\draw(0,0) node[tre] (1) {};
\draw (1) node {\tiny $1$};
\draw(1,0) node[tre] (2) {};
\draw (2) node {\tiny $2$};
\draw(2,0) node[tre] (3) {};
\draw (3) node {\tiny $3$};
\draw(3,0) node[tre] (4) {};
\draw (4) node {\tiny $4$};

\draw(0.2,1) node[small] (f) {};
\draw(2.7,1.5) node[small] (d) {};
\draw (3,1.5) node {\tiny $u_1$};
\draw(1.5,2.5) node[small] (b) {};
\draw (1.2,2.5) node {\tiny $u_2$};
\draw(2.5,2.5) node[small] (c) {};
\draw (2.8,2.4) node {\tiny $u_3$};
\draw(2,2) node[small] (e) {};
\draw(1.75,3.5) node[small] (a) {};
\draw (2.1,3.5) node {\tiny $u_4$};
\draw(1,4) node[small] (r) {};
\draw (1.3,4) node {\tiny $u_5$};
\draw(1,4.5) node[small] (newr) {};
\draw (1.5,4.5) node {\tiny $N$};

\draw[](newr)--(r);
\draw[](r)--(a);
\draw[](r)--(f);
\draw[](a)--(b);
\draw[](a)--(c);
\draw[](b)--(2);
\draw[](b)--(e);
\draw[](c)--(e);
\draw[](c)--(d);
\draw[](e)--(3);
\draw[](d)--(f);
\draw[](d)--(4);
\draw[](f)--(1);
\end{tikzpicture}
\qquad
 \begin{tikzpicture}[scale=0.8]
\draw(0,0) node[tre] (1) {};
\draw (1) node {\tiny $1$};
\draw(1,0) node[tre] (2) {};
\draw (2) node {\tiny $2$};
\draw(2,0) node[tre] (3) {};
\draw (3) node {\tiny $3$};
\draw(3,0) node[tre] (4) {};
\draw (4) node {\tiny $4$};

\draw(1.5,2.5) node[small] (b) {};
\draw(2.5,2.5) node[small] (c) {};
\draw(2,2) node[small] (e) {};
\draw(1.75,3.5) node[small] (a) {};
\draw(1,4) node[small] (r) {};
\draw(1,4.5) node[small] (newr) {};
\draw (2.5,4.5) node {\tiny $N_1=N^{(1,4)}$};

\draw[](newr)--(r);
\draw[](r)--(a);
\draw[](r)--(1);
\draw[](a)--(b);
\draw[](a)--(c);
\draw[](b)--(2);
\draw[](b)--(e);
\draw[](c)--(e);
\draw[](c)--(4);
\draw[](e)--(3);
\end{tikzpicture}
\qquad
\begin{tikzpicture}[scale=0.8]
\draw(0,0) node[tre] (1) {};
\draw (1) node {\tiny $1$};
\draw(1,0) node[tre] (2) {};
\draw (2) node {\tiny $2$};
\draw(2,0) node[tre] (3) {};
\draw (3) node {\tiny $3$};
\draw(3,0) node[tre] (4) {};
\draw (4) node {\tiny $4$};

\draw(1.5,2.5) node[small] (b) {};
\draw(1.75,3.5) node[small] (a) {};
\draw(1,4) node[small] (r) {};
\draw(1,4.5) node[small] (newr) {};
\draw (2.5,4.5) node {\tiny $N_2=N_1^{(3,4)}$}; 

\draw[](newr)--(r);
\draw[](r)--(a);
\draw[](r)--(1);
\draw[](a)--(b);
\draw[](a)--(4);
\draw[](b)--(2);
\draw[](b)--(3);
\end{tikzpicture}
\\
\begin{tikzpicture}[scale=0.8]
\draw(0,0) node[tre] (1) {};
\draw (1) node {\tiny $1$};
\draw(2,0) node[tre] (3) {};
\draw (3) node {\tiny $3$};
\draw(3,0) node[tre] (4) {};
\draw (4) node {\tiny $4$};

\draw(2.5,1) node[small] (a) {};
\draw(1.5,2) node[small] (r) {};
\draw(1.5,2.5) node[small] (newr) {};
\draw (3,2.5) node {\tiny $N_3=N_2^{(2,3)}$}; 

\draw[](newr)--(r);
\draw[](r)--(a);
\draw[](r)--(1);
\draw[](a)--(3);
\draw[](a)--(4);
\end{tikzpicture}
\qquad
\begin{tikzpicture}[scale=0.8]
\draw(0,0) node[tre] (1) {};
\draw (1) node {\tiny $1$};
\draw(2,0) node[tre] (4) {};
\draw (4) node {\tiny $4$};

\draw(1,1) node[small] (r) {};
\draw(1,1.5) node[small] (newr) {};
\draw (1,1.8) node {\tiny $N_4=N_3^{(3,4)}$};  

\draw[](newr)--(r);
\draw[](r)--(1);
\draw[](r)--(4);
\end{tikzpicture}
\qquad
\begin{tikzpicture}[scale=0.8]
\draw(0,0) node[tre] (4) {};
\draw (4) node {\tiny $4$};
\draw(0,0.5) node[small] (newr) {};
\draw (0,0.9) node {\tiny $N_5=N_4^{(1,4)}=I_4$};
\draw[](newr)--(4);
\end{tikzpicture}
 \end{tabular}

 \caption{
    An orchard network $N$ and the set of intermediate networks obtained by (cherry or reticulated-cherry) reductions until reaching $N^S$, the reduction of $N$ with respect to the complete sequence $S=(1,4)(3,4)(2,3)(3,4)(1,4)$. Note that, for example, the pair $(1,4)$ is a reticulated-cherry in $N$ and the pair $(2,3)$ is a cherry in $N_2$.}
 \label{fig:definicions}
\end{figure}

\section{Extended $\mu$-representation of a network}
\label{sec:exended-mu}

In this section we extend the definition of $\mu$-representation from~\cite{cardona2008comparison} and define a reduction operation that reduces certain building blocks of such data. In the next section, we shall see how this extension allows us to characterize orchard networks, in analogous way that the original definition allowed to compare tree-child networks.

\subsection*{Original $\mu$-representations and ancestral profiles}

We first recall the definitions of $\mu$-representation~\cite{cardona2008comparison} and of ancestral profile~\cite{erdos2019class}, but with some slight changes in order to homogenize notations with the rest of the paper.

Let $N=(V,A)$ be a network over $X\subseteq[n]$. For each node $u\in V$ and $i\in[n]$, we define $\mu_i(u)$ as the number of paths that start in $u$ and end in the leaf $i$, that is,
\[\mu_i(u)=m(u,i).\]
Obviously, if $i\notin X$, then $\mu_i(u)=0$.
The \defit{$\mu$-vector} of $u$ is the vector
\[\mu(u)=(\mu_1(u),\dots,\mu_n(u)).\]
In case that the network is not clear from the context, we will specify it writing $\mu^N(u)=(\mu^N_1(u),\dots,\mu^N_n(u))$.

The \defit{$\mu$-representation} of $N$ is the (multi)set
\[\bmu(N)=\{\mu(u) \mid u\in V_T\}.\] 
Note that in \cite{cardona2008comparison}, all nodes were taken into account, that is, $\bmu(N)$ also contained the $\mu$-vectors of the reticulations, but since the $\mu$-vector of such a node is equal to that of its single child, then both representations are equivalent.

Take now an arbitrary but fixed ordering of the tree nodes, say $v_1,\dots,v_t$. For every $i\in X$, consider its \emph{ancestral tuple} $\sigma(i)=(m(v_1,i),\dots,m(v_t,i))$, and let $\Sigma_N=\{(i,\sigma(i))\mid i\in X\}$ be the \emph{ancestral profile} of the network. 

Note that we can freely reorder the tree nodes (that is, applying a permutation to the nodes we get isomorphic networks), which amounts to applying the same permutation to the entries of all the ancestral tuples, but the sequence of ancestrals tuples cannot be reordered (that is, a permutation of the leaves gives non-isomorphic networks). Hence, mathematically, the object we get is an (ordered) sequence of (ordered) sequences, but on the latter the symmetric group acts simultaneously on all of them.



This object is equivalent to the $\mu$-representation, since both of them count the number of paths between tree nodes and leaves, but in the case of the $\mu$-representation the object is simpler, since we get a set, whose elements are (ordered) sequences, with no further action. 


The $\mu$-representation of a network (or, equivalently, its ancestral profile) is known to characterize different kinds of phylogenetic networks. Namely, tree-child~\cite{cardona2008comparison}, semi-binary tree-sibling time-consistent~\cite{cardona2008distance}, and stack-free orchard~\cite{Bai2021} networks.

\subsection*{Extended $\mu$-vectors}



We extend the definition of $\mu$-vectors by defining
$\mu_0(u)$ as the number of paths that start in $u$ and end in a reticulation of $N$, that is,
\[\mu_0(u)=\sum_{h\in V_H}m(u,h).\]
Clearly, if $N$ has no reticulations, then $\mu_0(u)=0$.
Then, the \defit{extended $\mu$-vector} of $u$ is the vector
\[\mu(u)=(\mu_0(u),\mu_1(u),\dots,\mu_n(u)).\]
As before, in case that the network is not clear from the context, we will specify it writing $\mu^N(u)=(\mu^N_0(u),\mu^N_1(u),\dots,\mu^N_n(u))$.

The \defit{extended $\mu$-representation} of $N$ is the (multi)set
\[\bmu(N)=\{\mu(u) \mid u\in V_T\}.\] 

Table \ref{tab:mudata} shows the extended $\mu$-representation of the network $N$ shown in Figure \ref{fig:definicions}. 

\begin{table}[ht]
    \centering
\begin{tabular}{cc} 
 \toprule
$u$  & $\mu(u)$ \\ 
 \midrule
 $1$ & $(0,1,0,0,0)$ \\
$2$ & $(0,0,1,0,0)$\\ 
$3$ & $(0,0,0,1,0)$\\ 
$4$ & $(0,0,0,0,1)$\\  
$u_1$ & $(1,1,0,0,1)$\\ 
$u_2$ & $(1,0,1,1,0)$\\ 
$u_3$ & $(2,1,0,1,1)$\\ 
$u_4$ & $(3,1,1,2,1)$\\ 
$u_5$ & $(4,2,1,2,1)$\\ 
 \bottomrule
\end{tabular}
    \caption{Extended $\mu$-representation of the network $N$ in Figure \ref{fig:definicions}.}
    \label{tab:mudata}
\end{table}

Sometimes we shall omit the word ``extended'' and simply talk of $\mu$-vector and $\mu$-representation. In case we refer to the original definition~\cite{cardona2008comparison} (where the number of paths to reticulations is not counted), we shall explicitly say so.
Note also that we have defined $\bmu(N)$ as a multiset, in which every vector $\mu$ appears as many times as the number of internal tree nodes having it as its $\mu$-vector. We shall later see that in the cases we are interested in, no pair of different tree nodes can have the same $\mu$-vector, and hence we can consider it a set. For now, we keep considering it a multiset.

For any $S\subseteq [n]^*$, let $\delta_S$ be the vector whose $i$-th component (indexed from $0$ to $n$) is $1$ if $i\in S$ and $0$ if $i\in[n]^*\setminus S$. To simplify the notation, for any $j_1,\dots,j_k\in[n]^*$, we define $\delta_{j_1,\dots,j_k}=\delta_{\{j_1,\dots,j_k\}}$; in particular, if $j\in[n]^*$, then $\delta_{j}=\delta_{\{j\}}$. Note that $\bmu(I_i)=\{\delta_i\}$, where we recall that $I_i$ is the trivial tree whose only leaf is $i$.

The following proposition shows that the $\mu$-representation of a network $N=(V,A)$ on $X\subseteq[n]$ can be computed sequentially, for instance, visiting the nodes in order of increasing height (which is the length of a longest path starting at the node, and can be computed in $O(|V|)$ time). As a consequence, the $\mu$-representation of such a network  can be computed in $O(|V|n)$ time.

\begin{prop}\label{formula-for-mu}
    Let $u\in V$. Then:
    \begin{itemize}
        \item If $u$ is the leaf $i$, then $\mu(u)=\delta_i$.
        \item If $u$ is a tree node, then $\mu(u)=\mu(v)+\mu(w)$, where $v,w$ are the children of $u$.
        \item If $u$ is a reticulation, then $\mu(u)=\mu(v)+\delta_0$, where $v$ is the single child of $u$.
    \end{itemize}
\end{prop}

\begin{proof}
    It is clear that the only paths that start in a leaf are trivial, and end in the same leaf. This proves the first assertion.

    If $u$ is a tree node, with children $v$ and $w$, then for each leaf $i$, the set of paths in $M(u,i)$ is in bijection with the disjoint union of the paths in $M(v,i)$ and $M(w,i)$, by prepending or deleting the arc $uv$ and $uw$ respectively. Since $u$ is not a reticulation, this same remark holds for the paths leading to reticulations.

    If $u$ is a reticulation and $v$ is its single child, then for each leaf $i$, the set of paths in $M(u,i)$ is in bijection with the set of paths in $M(v,i)$, by prepending or deleting the arc $uv$. As for paths from $u$ to reticulations, apart from those that are obtained from the paths starting at $v$ by prepending the arc $uv$, there is a new (trivial) path starting (and ending) in $u$. 
\end{proof}

The relation ``being descendant of'' gives to the set of nodes of a network the structure of partially ordered set. Also, the $\mu$-vectors of nodes admit such a structure, namely saying that $\mu(v)\le \mu(u)$ if and only if $\mu_i(v)\le\mu_i(u)$ for every $i\in[n]^*$. The following technical lemma, which shall be used later, shows how these two partial orders are related.

\begin{lem}\label{paths_and_mu}
    Let $u,v$ be two nodes of a network $N$. 
    \begin{itemize} 
        \item If $v$ is descendant of $u$, then $\mu(v)\le \mu(u)$.
        \item If, moreover, $u$ and $v$ are different tree nodes, then $\mu(u)\neq \mu(v)$.
    \end{itemize}
\end{lem}

\begin{proof}
    The first statement follows easily from Proposition~\ref{formula-for-mu}.

    For the second one, suppose that there is a non-trivial path from a tree node $u$ to another tree node $v$ and $\mu(u)=\mu(v)$. From the first part of the statement, all nodes in this path have the same $\mu$ vector. Let $w\neq u$ the first visited node in this path (and note that $w$ is a child of $u$). 
    If $w$ is a reticulation (and notice that $w\neq v$ since $v$ is a tree node), consider its single child $w'$; then, $\mu_0(w)=\mu_0(w')+1>\mu_0(w')$ and we get $\mu_0(u)\ge\mu_0(w)>\mu_0(w')\ge\mu_0(v)$, which is a contradiction.  If $u$ is a tree node, and $w'$ is the child of $u$ different from $w$, then $\mu(u)=\mu(w)+\mu(w')$ and again we get a contradiction since $\mu(w')\neq 0$.
\end{proof}


\subsection*{Cherries and reticulated-cherries in $\mu$-representations}

Now we mimic the definitions of cherry and reticulated-cherry, which were introduced for pairs of leaves of a network, to be applied over pairs of indices in a (multi)set of vectors, and we prove that they make sense, since taking the $\mu$-representation of a network transforms the former into the latter.

Let $\bmu$ be a (multi)set formed by elements that are vectors of length $n+1$, whose entries are non-negative integers and are indexed by $0,1,\dots,n$. We say that $(i,j)$ is a \defit{cherry} of $\bmu$ if:    
\begin{enumerate}
    \item $\delta_{i,j}\in\bmu$ (with multiplicity $1$ in the multiset) and
    \item for each $\mu=(\mu_0,\mu_1,\dots,\mu_n)\in \bmu\setminus\{\delta_i,\delta_j\}$, we have that $\mu_i=\mu_j$.
\end{enumerate} 

\begin{prop}
    \label{cherry-in-terms-of-mu}
    Let $i,j\in[n]$. Then, $(i,j)$ is a cherry of $N$ if, and only if, it is a cherry of $\bmu(N)$.
\end{prop}

\begin{proof}
    Assume that $(i,j)$ is a cherry of $N$, and let $p$ be its common parent. Clearly, $p\in V_T$ and $\mu(p)=\delta_{i,j}$ belongs to $\bmu(N)$. 

    Assume that $p'\neq p$ is a tree node with $\mu(p')=\delta_{i,j}$. Since all the paths from the root $\rho$ to $j$ must pass through $p$, it follows that $p'$ is either a descendant or an ascendant of $p$. From Lemma~\ref{paths_and_mu} it follows that $p=p'$, which is a contradiction.

    For each tree node $u$, the set of paths in $M(u,i)$ (resp. $M(u,j)$) is in bijection with the set of paths in $M(u,p)$, by appending and removing the arc $\arc pi$ (resp. $\arc pj$); hence $\mu_i(u)=m(u,i)=m(u,p)=m(u,j)=\mu_j(u)$. Finally, if $u$ is a leaf different from $i,j$, then clearly $\mu_i(u)=\mu_j(u)=0$.

    Conversely, assume that $(i,j)$ is a cherry of $\bmu(N)$, and let $p\in V_T$ with $\mu(p)=\delta_{i,j}$. Since $\delta_0(p)=0$, there is no reticulation reachable from $p$, which means that the subnetwork rooted at $p$ must be formed by tree nodes and hence it is a tree without elementary nodes whose only leaves are $i,j$ (since for each other $k\in[n]$, $\delta_k(p)=0$). Hence, this subtree must be a cherry.
\end{proof}

We say that $(i,j)$ is a \defit{reticulated-cherry} of $\bmu$ if:
\begin{enumerate}
    \item $\delta_{0,i,j}\in\bmu$ (with multiplicity $1$ in the multiset) and
    \item for each $\mu=(\mu_0,\mu_1,\dots,\mu_n)\in \bmu\setminus\{\delta_i,\delta_j\}$, we have that $\mu_0\ge\mu_i\ge\mu_j$.
\end{enumerate}

\begin{prop}
    \label{ret-cherry-in-terms-of-mu}
    Let $i,j\in[n]$, then $(i,j)$ is a reticulated-cherry of $N$ if, and only if, it is a reticulated-cherry of $\bmu(N)$.
\end{prop}

\begin{proof}
    Assume that $(i,j)$ is a reticulated-cherry of $N$, and let $p,q$ be the respective parents of $i,j$, which means that $p$ is a reticulation, $q$ is a tree node, and there is an arc $\arc qp$. It is clear that $\mu(q)=\delta_{0,i,j}$. 
    
    Assume that $q'\neq q$ is a tree node with $\mu(q)=\delta_{0,i,j}$. Since all the paths from $\rho$ to $j$ must pass through $q$, it follows that $q'$ is either a descendant or an ascendant of $q$. From Lemma~\ref{paths_and_mu} it follows that $q=q'$, which is a contradiction.
    
    Let $u$ be any tree node of $N$. First, note that the paths in $M(u,i)$ are in bijection with the paths in $M(u,p)$, by appending and removing the arc $\arc pi$, hence $m(u,i)=m(u,p)$. Since $p$ is a reticulation, it follows that $\mu_0(u)\ge m(u,p)=m(u,i)=\mu_i(u)$. Second, note that all paths in $M(u,j)$ must pass through $q$, while those in $M(u,i)$ can either pass through $q$ or through the other parent of $p$ different from $q$. From this it follows that $\mu_i(u)=m(u,i)\ge m(u,j)=\mu_j(u)$ and, finally, the condition $\mu_0(u)\ge\mu_i(u)\ge\mu_j(u)$ holds.

    If $u\neq i,j$ is a leaf of $N$, it is clear that $\mu_0(u)=\mu_i(u)=\mu_j(u)=0$ and also the condition $\mu_0(u)\ge\mu_i(u)\ge\mu_j(u)$ holds.

    Conversely, assume that $(i,j)$ is a reticulated-cherry of $\bmu(N)$, and let ${u}\in V_T$ be such that $\mu({u})=\delta_{0,i,j}$. Since  the leaves reachable from $u$ are $i,j$, the subnetwork of $N$ rooted at $u$ must be a tree, with elementary nodes that are reticulations in $N$. Since $\mu_0(u)=1$, there is a single path from $u$ to a reticulation, which means that there can only be a reticulation reachable from $u$. Then, the subnetwork rooted at $u$ must be a cherry with a single elementary node inserted above one of its leaves. Let $p,q$ be the respective parents of $i,j$. Then, either $u=q$ and there exists an arc $\arc qp$, or $u=p$ and the arc is $\arc pq$. In the second case, let $v$ the single child of the root of $N$, which necessarily belongs to $V_T$; all the paths in $M(v,i)$ must pass through $p$, while those ending in $j$ can pass through $p$ or through the other parent of $q$ different from $p$ (and there is at least one of these paths); from this it follows that $\mu_j(v)>\mu_i(v)$, against the hypothesis. Then, there exists an arc $\arc qp$ and $(i,j)$ is a reticulated-cherry.
\end{proof}


\begin{coro}
    \label{cherry-and-ret-cherry-incompatible}
    If $\bmu$ is the $\mu$-representation of an orchard network, then the conditions that $(i,j)$ is a cherry and a reticulated-cherry of $\bmu$ are mutually exclusive.
\end{coro}

We say that $(i,j)$ is \defit{reducible} in $\bmu$ if it is either a cherry or a reticulated-cherry. With this notation, Propositions~\ref{cherry-in-terms-of-mu} and~\ref{ret-cherry-in-terms-of-mu} can be written as follows.

\begin{coro}
\label{pair-reducible-in-mu-iff-in-net}
    Let $i,j\in[n]$. Then, $(i,j)$ is reducible in $N$ if, and only if, it is reducible in $\bmu(N)$.
\end{coro}

\subsection*{Reduction of $\mu$-representations}

We now extend the concept of reduction, from networks to (multi)sets of vectors. As before, this extension is compatible with the computation of the $\mu$-representation of a network.

Let $\bmu$ be a (multi)set of vectors of length $n+1$ whose entries are non-negative integers, and assume that $(i,j)$ is a cherry of $\bmu$. We define the \defit{cherry-reduction} of $\bmu$ with respect to $(i,j)$ as the (multi)set obtained from $\bmu$ by means of the following operations:
    \begin{itemize}
        \item Remove $\delta_i$ and $\delta_{i,j}$ from $\bmu$;
        \item For every other $\mu=(\mu_0,\mu_1,\dots,\mu_n)\in\bmu$, set $\mu_i=0$.
    \end{itemize}
    More formally, it is the (multi)set:
    \[\{ \mu-\mu_i\delta_i\mid \mu=(\mu_0,\mu_1,\dots,\mu_n)\in\bmu\setminus\{\delta_i,\delta_{i,j}\}\}.\]

Let $\bmu$ be a (multi)set of vectors of length $n+1$ whose entries are non-negative integers, and assume that $(i,j)$ is a reticulated-cherry of $\bmu$. We define the \defit{reticulated-cherry-reduction} of $\bmu$ with respect to $(i,j)$ as the (multi)set obtained from $\bmu$ by means of the following operations:
    \begin{itemize}
        \item Remove $\delta_{0,i,j}$ from $\bmu$;
        \item For every other $\mu=(\mu_0,\mu_1,\dots,\mu_n)\in\bmu$, subtract $\mu_i$ to $\mu_0$, and $\mu_j$ to $\mu_i$.
    \end{itemize}
    More formally, it is the (multi)set:
    \[\{ \mu-\mu_i\delta_0-\mu_j\delta_i\mid \mu=(\mu_0,\mu_1,\dots,\mu_n)\in\bmu\setminus\{\delta_{0,i,j}\}\}.
        \]


In case that the pair $(i,j)$ is not simultaneously a cherry and also a reticulated-cherry in $\bmu$ (and recall that, thanks to Corollary~\ref{cherry-and-ret-cherry-incompatible}, this will always be the case when dealing with the $\mu$-representation of an orchard network) we shall simply call it the \defit{reduction} of $\bmu$ with respect to $(i,j)$ and denote it by $\bmu^{(i,j)}$.
Note also that the process of modifying the $\mu$-vectors could in principle lead to repeated elements (even when there are no repeated elements in $\bmu$), but since we consider $\bmu$ as a multiset, we can take into account these repetitions. As noted before, it will never happen for the cases we shall be interested in.


\begin{prop}
    \label{mu-of-reductions}
    Let $(i,j)$ be a reducible pair in a network $N$. Then,
    \[\bmu(N^{(i,j)})=\bmu(N)^{(i,j)}.\]
\end{prop}

\begin{proof}
    Let us first assume that $(i,j)$ is a cherry in $N$, which implies that $(i,j)$ is a cherry  in $\bmu(N)$. For simplicity, we write $\tilde N=N^{(i,j)}$.

    Note that $p$, the parent of both $i$ and $j$, belongs to $V_T$ and is the only node that becomes elementary when we suppress the node $i$. Clearly $\mu(p)=\delta_{i,j}$ will not contribute to $\bmu(\tilde N)$. For any other node $u$ of $N$, we have that $u$ also belongs to $\tilde N$. Moreover, $\mu_k^{\tilde N}(u)= m_{\tilde N}(u,k)=m_N(u,k)\mu_k^{N}(u)$ 
    (for every $k\in[n]\setminus\{i\}$) since the paths in $M_N(u,k)$ are in bijection with the paths in $M_{\tilde N}(u,k)$, while $\mu_i^{\tilde N}(u)=m_{\tilde N}(u,i)=0$, since $i$ is no longer a leaf of $\tilde N$.  
  Therefore, the $\mu$-representation of $\tilde N$ coincides with the reduction of the  $\mu$-representation of $N$ with respect to $(i,j)$.
    
    Let us now assume that $(i,j)$ is a reticulated-cherry in $N$, which implies that $(i,j)$ is a reticulated-cherry  in $\bmu(N)$. As before, we write $\tilde N=N^{(i,j)}$.

    Note that $p,q$, the respective parents of $i,j$ become elementary when the arc that connects them is removed. However, $p$ is a reticulation in $N$, and hence its $\mu$-vector does not contribute to $\bmu(N)$. Then, $q$, which belongs to $V_T$ and whose $\mu$-vector in $N$ is $\delta_{0,i,j}$, is the only node whose contribution to $\bmu(N)$ disappears in $\bmu(\tilde N)$. 
    
    For any other node $u$ of $N$, we have that $u$ also belongs to $\tilde N$. The paths in $M_N(u,i)$ can be divided between those that pass through $q$ (which are in bijection with $M_N(u,j)$) and those that pass through the other parent of $p$. The former will no longer be paths in $\tilde N$ when we remove the arc $qp$, while the latter will be in bijection with the paths in $M_{\tilde N}(u,i)$. Hence, we get $\mu^{\tilde N}_{i}(u)=m_{\tilde N}(u,i)=m_N(u,i)-m_N(u,j)=\mu^N_i(u)-\mu^N_j(u)$.
    Note also that $p$ is the only reticulation in $N$ that disappears in $\tilde N$, and the paths in $M_N(u,p)$ are in bijection with those in $M_N(u,i)$. Hence, $\mu^{\tilde N}_{0}(u)=\mu^N_{0}(u)-m_N(u,p)=\mu^N_{0}(u)-\mu^N_{i}(u)$. Therefore, and as in the previous case, the operations we have to perform to obtain $\bmu(\tilde N)$ from $\bmu(N)$ are exactly the same than those we have to do to compute $\bmu(N)^{(i,j)}$.
\end{proof}

\section{Classification of orchard networks using $\mu$-representations}
\label{sec:mu-for-orchard}

In the last section, we have introduced the $\mu$-rep\-re\-sentation associated to a network, and the reductions that can eliminate cherries and reticulated cherries. Only accidentally we have referred to the case of orchard networks. In this section we show how the $\mu$-representation can be used to characterize orchard networks up to isomorphism and to define a sound distance measure among them.

\subsection*{Unicity of $\mu$-vectors in orchard networks}

We have already announced (without proof) that, although the $\mu$-rep\-re\-sentation of a network has been defined as a multiset, in the case of orchard networks, there are no repeated $\mu$-vectors and hence it can be considered a set. We now prove this result.

\begin{prop}
    Let $N$ be an orchard network, and $u,v\in V_T(N)$. If $\mu(u)=\mu(v)$, then $u=v$.
\end{prop}

\begin{proof}
    The result is clear if either $u$ or $v$ are leaves.
    
    Let $S$ be a complete reducible sequence for $N$, and consider the intermediate networks $N_k=(\dots(N^{s_1})^{s_2}\dots)^{s_k}$. At some point in this reduction process, each of $u$ and $v$ are eventually simplified because they become elementary nodes. Let us assume that $u,v$ belong to $N_{k}$, but $u$ does not belong to $N_{k+1}$. If we let $(i,j)=s_{k+1}$, then either $\mu^{N_{k}}(u)=\delta_{i,j}$ or $\mu^{N_{k}}(u)=\delta_{0,i,j}$, depending on whether $s_{k+1}$ is a cherry or a reticulated-cherry. Now, since $\bmu({N_k})=(\dots(\bmu(N)^{s_1})^{s_2}\dots)^{s_k}$ we have that $\mu^{N_{k}}(v)=\mu^{N_{k}}(u)$. If $u\neq v$, then the multiplicity of $\delta_{i,j}$ (or $\delta_{0,i,j}$)  in $\bmu({N_{k}})$ would be at least $2$, against the hypothesis that $s_{k+1}$ is a reducible pair for $\bmu({N_{k}})$.
\end{proof}

\begin{coro}
    The $\mu$-representation $\bmu(N)$ of an orchard network $N$ is a set.
\end{coro}

\subsection*{Reducible sequences and isomorphism of orchard networks}

We now show how the definitions and results introduced so far allow us to prove that the $\mu$-representation of a network classifies with unicity its isomorphism class. To do so, we translate into $\mu$-representations the concepts of complete reducible sequences of networks, and use the classification result of Theorem~\ref{same-crs-implies-iso}.

Given a sequence of pairs of integers $S=s_1\cdots s_k$, with $s_t=(i_t,j_t)$, of length $k\ge1$, we say that $S$ is \defit{reducible} in $\bmu$ if:
\begin{itemize}
    \item $s_1$ is reducible in  $\bmu$.
    \item For every $t\in 2,\dots,k$, $s_t$ is reducible in
    $(\dots(\bmu^{s_1})^{s_2}\dots)^{s_{t-1}}$.
\end{itemize}
In such a case, we shall define the \defit{reduction} of $\bmu$ with respect to $S$ as $(\dots(\bmu^{s_1})^{s_2}\dots)^{s_{k}}$ and will denote it as $\bmu^S$.

\begin{thm}
    A sequence of pairs is reducible in an orchard network $N$ if, and only if, it is so for $\bmu(N)$. Moreover, in such a case,
    \[\bmu(N^S)=\bmu(N)^S.\]
\end{thm}

\begin{proof}
    If follows from the iterative application of Corollary~\ref{pair-reducible-in-mu-iff-in-net} and Proposition~\ref{mu-of-reductions}.
\end{proof}

We say that $S$, a reducible sequence in $\bmu$, is \defit[reducible sequence]{complete} if $\bmu^S=\{\delta_i\}$ for some integer $i$.

\begin{coro}
    \label{comp-red-seq-net-mu}
    A reducible sequence of $N$ is complete if, and only if, it is so for $\bmu(N)$.
\end{coro}

\begin{thm}
\label{mu-classifies}
    Let $N,N'$ be two orchard networks on $X,X'\subseteq[n]$. Then, $N\cong N'$ if, and only if,
    $\bmu(N)=\bmu(N')$.
\end{thm}
\begin{proof}
    The ``only if'' part of the statement is clear. Conversely, assume that $\bmu(N)=\bmu(N')$, and let $S$ be a complete reducible sequence for $N$. Then, thanks to Corollary~\ref{comp-red-seq-net-mu}, $S$ is also a complete reducible sequence for $\bmu(N)$, and since $\bmu(N)=\bmu(N')$, applying again Corollary~\ref{comp-red-seq-net-mu}, it follows that $S$ is also a complete reducible sequence for $N'$. Finally, applying Theorem~\ref{same-crs-implies-iso} we get that $N\cong N'$.
\end{proof}

We finish this section noticing that $\mu$-representations have greater separation power than the \emph{original} $\mu$-representations (or, equivalently, ancestral profiles). For instance, the networks in Figure~\ref{fig:same_ancestral} (adapted from  \cite[Figure~2]{Bai2021}) have the same \emph{original} $\mu$-representation (or ancestral profile), but different \emph{extended} $\mu$-representation. Indeed, the internal tree nodes have been ordered so that the ancestral tuples of all leaves are the same, but $\mu^{N_1}(v_4)=(2,0,1,0,1)$, $\mu^{N_1}(v_5)=(1,0,0,1,1)$, while $\mu^{N_2}(v_4)=(1,0,1,0,1)$, $\mu^{N_2}(v_5)=(2,0,0,1,1)$.

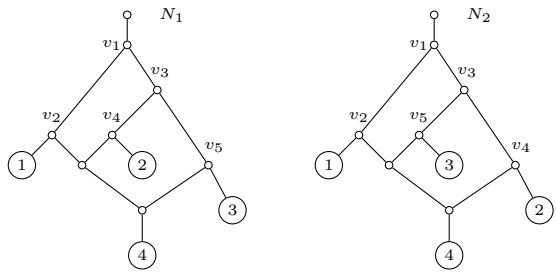
\begin{figure}[t]
\centering
\begin{tabular}{@{}c@{}}
\begin{tikzpicture}[scale=0.8]
\draw(0,1.5) node[tre] (1) {};
\draw (1) node {\tiny $1$};
\draw(2,1.5) node[tre] (2) {};
\draw (2) node {\tiny $2$};
\draw(3.5,0.75) node[tre] (3) {};
\draw (3) node {\tiny $3$};
\draw(2,0) node[tre] (4) {};
\draw (4) node {\tiny $4$};
\draw(0.5,2) node[small] (v2) {};
\draw (0.5,2.3) node {\tiny $v_2$};
\draw(1.5,2) node[small] (v4) {};
\draw (1.5,2.3) node {\tiny $v_4$};
\draw(3.1,1.5) node[small] (v5) {};
\draw (3.2,1.8) node {\tiny $v_5$};
\draw(2.25,2.75) node[small] (v3) {};
\draw (2.3,3.05) node {\tiny $v_3$};
\draw(1.75,3.5) node[small] (v1) {};
\draw (1.5,3.5) node {\tiny $v_1$};
\draw(1,1.5) node[small] (r1) {};
\draw(2,0.75) node[small] (r2) {};

\draw[](v1)--(v2);
\draw[](v1)--(v3);
\draw[](v2)--(1);
\draw[](v2)--(r1);
\draw[](v3)--(v4);
\draw[](v3)--(v5);
\draw[](v4)--(2);
\draw[](v4)--(r1);
\draw[](r1)--(r2);
\draw[](v5)--(3);
\draw[](v5)--(r2);
\draw[](r2)--(4);
\draw (1.75,4) node[small] (rho) {};
\draw (2.5,4) node {\tiny $N_1$};
\draw[](rho)--(v1);
\end{tikzpicture}
\qquad
\begin{tikzpicture}[scale=0.8]
\draw(0,1.5) node[tre] (1) {};
\draw (1) node {\tiny $1$};
\draw(2,1.5) node[tre] (2) {};
\draw (2) node {\tiny $3$};
\draw(3.5,0.75) node[tre] (3) {};
\draw (3) node {\tiny $2$};
\draw(2,0) node[tre] (4) {};
\draw (4) node {\tiny $4$};

\draw(0.5,2) node[small] (v2) {};
\draw (0.5,2.3) node {\tiny $v_2$};
\draw(1.5,2) node[small] (v4) {};
\draw (1.5,2.3) node {\tiny $v_5$};
\draw(3.1,1.5) node[small] (v5) {};
\draw (3.2,1.8) node {\tiny $v_4$};
\draw(2.25,2.75) node[small] (v3) {};
\draw (2.3,3.05) node {\tiny $v_3$};
\draw(1.75,3.5) node[small] (v1) {};
\draw (1.5,3.5) node {\tiny $v_1$};
\draw(1,1.5) node[small] (r1) {};
\draw(2,0.75) node[small] (r2) {};

\draw[](v1)--(v2);
\draw[](v1)--(v3);
\draw[](v2)--(1);
\draw[](v2)--(r1);
\draw[](v3)--(v4);
\draw[](v3)--(v5);
\draw[](v4)--(2);
\draw[](v4)--(r1);
\draw[](r1)--(r2);
\draw[](v5)--(3);
\draw[](v5)--(r2);
\draw[](r2)--(4);
\draw (1.75,4) node[small] (rho) {};
\draw (2.5,4) node {\tiny $N_2$};
\draw[](rho)--(v1);
\end{tikzpicture}
 \end{tabular}
 \caption{Non-isomorphic networks with the same \emph{original} $\mu$-representation but different \emph{extended} $\mu$-representation.}
 \label{fig:same_ancestral}
\end{figure}



\subsection*{A sound distance on the class of orchard networks}
\label{sec:mu-distance}

As a by-product, we get that $\mu$-representations can be used to define a sound distance measure on the set of (isomorphims classes of) orchard networks. We recall that defining such a distance amounts to giving a mapping that assigns to every pair of networks $N_1,N_2$ a real number $d(N_1,N_2)$ with the properties:
\begin{enumerate}
    \item (positivity) $d(N_1,N_2)\ge 0$;
    \item (separation) $d(N_1,N_2)=0$ if, and only if, $N_1\cong N_2$;
    \item (symmetry) $d(N_1,N_2)=d(N_2,N_1)$;
    \item (triangular inequality) $d(N_1,N_3)\le d(N_1,N_2)+d(N_2,N_3)$.
\end{enumerate}

\begin{thm}
    The mapping
\[d_{\mu}(N_1,N_2):=|\bmu(N_1)\bigtriangleup \bmu(N_2)|,\]
defines a distance on the set of isomorphism classes of orchard networks over sets of taxa included in $[n]$.   
\end{thm}

\begin{proof}
    Since the mapping is based on the cardinality of the symmetric difference of sets, only the separation property needs to be checked, which is given by Theorem~\ref{mu-classifies}.
\end{proof}

Note that, as it happened in the case of the \emph{original} $\mu$-distance~\cite{cardona2008comparison}, the distance $d_{\mu}$ just defined, when restricted to phylogenetic trees, gives the classical Robinson-Foulds distance~\cite{robinson1981comparison}.
Note also that distances are usually defined between phylogenetic networks (or trees) over the same set of taxa, but $d_{\mu}$  can also be applied to networks over different set of taxa, as long as they belong to some common superset.

\section{Conclusions} 
\label{sec:conclusions}

In this paper, we have presented an extension of the $\mu$-representation of a network that allows us to separate and define distances on arbitrary orchard networks and is computationally efficient. Notice, however, that the $\mu$-representation cannot separate more general kinds of networks. For instance, the two networks in Figure~\ref{fig:false_for_others}, which are tree-sibling \cite{cardona2008distance}, stack-free \cite{semple2018phylogenetic}, separable \cite{pons2022polynomial}, FU-stable \cite{huber2016folding} and tree-based \cite{francis2015phylogenetic}, have the same $\mu$-representation but are not isomorphic. It should come as no surprise that invariants that are easy to compute and to compare cannot separate arbitrary networks since, as already observed, even for the (quite restrictive) case of tree-sibling time-consistent networks, its comparison is Graph Isomorphism Complete.

\begin{figure}[t]
\centering
\begin{tabular}{@{}c@{}}
\begin{tikzpicture}[scale=1]
\draw(2,4.5) node[small] (r) {};
\draw(2,5) node[small] (rho) {};
\draw[](rho)--(r);
\draw (2.4,5) node {\tiny $N$};

\draw(0,0) node[tre] (1) {};
\draw (1) node {\tiny $1$};
\draw(1,0) node[tre] (2) {};
\draw (2) node {\tiny $2$};
\draw(2,0) node[tre] (3) {};
\draw (3) node {\tiny $3$};
\draw(3,0) node[tre] (4) {};
\draw (4) node {\tiny $4$};
\draw(4,0) node[tre] (5) {};
\draw (5) node {\tiny $5$};

\draw(0,1) node[small] (A) {};
\draw(1,1) node[small] (B) {};
\draw(3,1) node[small] (C) {};
\draw(4,1) node[small] (D) {};

\draw(0,1.7) node[small] (a) {};
\draw(0.5,2.4) node[small] (b) {};
\draw(1,3.1) node[small] (c) {};
\draw(1.5,3.8) node[small] (d) {};
\draw(4,1.7) node[small] (e) {};
\draw(3.5,2.4) node[small] (f) {};
\draw(3,3.1) node[small] (g) {};

\draw[](A)--(1);
\draw[](B)--(2);
\draw[](C)--(4);
\draw[](D)--(5);
\draw[](r)--(d);
\draw[](d)--(c);
\draw[](c)--(b);
\draw[](b)--(a);
\draw[](a)--(A);
\draw[](r)--(g);
\draw[](g)--(f);
\draw[](f)--(e);
\draw[](e)--(D);
\draw[](a)--(B);
\draw[](b)--(C);
\draw[](c)--(D);
\draw[](d)--(3);
\draw[](e)--(C);
\draw[](f)--(B);
\draw[](g)--(A);
\end{tikzpicture}
\qquad
\begin{tikzpicture}[scale=1]
\draw(2,4.5) node[small] (r) {};
\draw(2,5) node[small] (rho) {};
\draw[](rho)--(r);
\draw (2.4,5) node {\tiny $N'$};

\draw(0,0) node[tre] (1) {};
\draw (1) node {\tiny $5$};
\draw(1,0) node[tre] (2) {};
\draw (2) node {\tiny $4$};
\draw(2,0) node[tre] (3) {};
\draw (3) node {\tiny $3$};
\draw(3,0) node[tre] (4) {};
\draw (4) node {\tiny $2$};
\draw(4,0) node[tre] (5) {};
\draw (5) node {\tiny $1$};

\draw(0,1) node[small] (A) {};
\draw(1,1) node[small] (B) {};
\draw(3,1) node[small] (C) {};
\draw(4,1) node[small] (D) {};

\draw(0,1.7) node[small] (a) {};
\draw(0.5,2.4) node[small] (b) {};
\draw(1,3.1) node[small] (c) {};
\draw(1.5,3.8) node[small] (d) {};
\draw(4,1.7) node[small] (e) {};
\draw(3.5,2.4) node[small] (f) {};
\draw(3,3.1) node[small] (g) {};

\draw[](A)--(1);
\draw[](B)--(2);
\draw[](C)--(4);
\draw[](D)--(5);
\draw[](r)--(d);
\draw[](d)--(c);
\draw[](c)--(b);
\draw[](b)--(a);
\draw[](a)--(A);
\draw[](r)--(g);
\draw[](g)--(f);
\draw[](f)--(e);
\draw[](e)--(D);
\draw[](a)--(B);
\draw[](b)--(C);
\draw[](c)--(D);
\draw[](d)--(3);
\draw[](e)--(C);
\draw[](f)--(B);
\draw[](g)--(A);
\end{tikzpicture}
 \end{tabular}
 \caption{Non-isomorphic networks with the same $\mu$-representation.}
 \label{fig:false_for_others}
\end{figure}
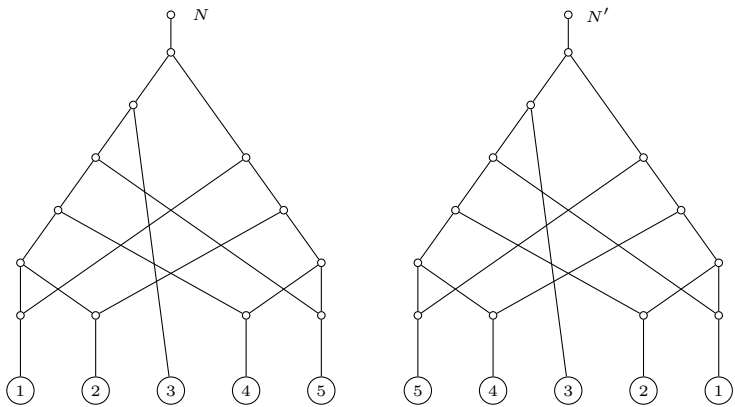

We have implemented the decomposition and reconstruction of orchard networks using reducible sequences, the computation of the $\mu$-representation of a network, the reconstruction of an orchard network given its $\mu$-represen\-ta\-tion, which recovers the original network in case it was orchard, and also the computation of the $\mu$-distance between networks. This implementation appears in version 2.2 of the Python package \verb|phylonetwork|, which is available from \url{https://pypi.org/project/phylonetwork/}. We have also made a demo of the aforementioned features, available at \url{https://github.com/gerardet46/OrchardMuRepresentation}, which can be run online with the link to \url{https://mybinder.org/} provided therein.

\bibliographystyle{elsarticle-num-names}
\bibliography{bibliography}

\end{document}